\documentclass[11pt]{article}

\usepackage{amsmath,amssymb,amsthm,mathtools, braket}
\usepackage{xcolor}
\usepackage[margin=1in]{geometry}
\usepackage[colorlinks=true,linkcolor=blue,citecolor=blue,urlcolor=blue]{hyperref}

\newcommand{\F}{\mathbb F}

\newcommand{\poly}{\mathrm{poly}}

\newcommand{\ba}{\boldsymbol{a}}
\newcommand{\bb}{\boldsymbol{b}}

\newcommand{\boldf}{\boldsymbol{f}}

\newcommand{\bi}{\boldsymbol{i}}

\newcommand{\bell}{\boldsymbol{\ell}}

\newcommand{\bu}{\boldsymbol{u}}
\newcommand{\bv}{\boldsymbol{v}}

\newtheorem{theorem}{Theorem}
\newtheorem{lemma}{Lemma}
\newtheorem{corollary}{Corollary}
\newtheorem{proposition}{Proposition}

\theoremstyle{definition}
\newtheorem{definition}{Definition}

\title{The Low-Individual-Degree Test Without the Diagonal-Lines Test Is Not Quantum-Sound}
\author{Honghao Fu\\Concordia University \and Nengkun Yu\\Stony Brook University \and Tianrun Zhao\\Stony Brook University}

\date{\today}

\begin{document}

\maketitle

\begin{abstract}
To prove the quantum soundness of the classical low-individual-degree test, the authors of \cite{JNVWY20LID} defined three subtests, namely the axis-parallel lines test, the self-consistency test, and the diagonal-lines test. An interesting question is whether the diagonal-lines test can be removed. In this paper, we show that the diagonal-lines test cannot simply be removed without another compatibility mechanism. Consequently, replacing the "conditional linear functions" by "coordinate deletion functions" in the proof of MIP*=RE, as mentioned in \cite{JNVWY20LID}, does not by itself preserve the required soundness.

The authors of \cite{JNVWY20LID} found an example that requires the diagonal-lines test when \((m, d, q) = (2, 2, 4)\); we give an example when \((m, d) = (2, 2)\) and \(q\) is any odd prime.
\end{abstract}

\section*{AI Disclosure}

The initial main idea was suggested by ChatGPT 5.5 Pro.  The authors are
responsible for the mathematical statements and proofs. This paper is manually written by the authors.

\section{Introduction}
\label{sec:introduction}

Low-degree tests are a basic tool for checking that local answers arise from
a single global polynomial.  Linearity self-testing was introduced by Blum,
Luby, and Rubinfeld~\cite{BLR93}.  Multilinearity and low-individual-degree
testing were introduced and proved sound by Babai, Fortnow, and Lund in their
work on \(\mathsf{MIP}=\mathsf{NEXP}\)~\cite{BFL91}.  Classical analyses of
these tests and closely related low-degree tests were developed further in
\cite{FHS94,PS94,RubSud96,RS97,AroraSudan03}, and became central ingredients
in the PCP theorem~\cite{AS98,ALMSS98}.  Axis-parallel tests have also been
studied in the broader setting of tensor-product codes~\cite{CMS20}.

Quantum soundness is substantially more delicate: answers associated with
different questions are produced by quantum measurements, and those
measurements need not be jointly measurable.  The study of quantum-sound
multilinearity tests began with Ito and Vidick~\cite{IV12}.  Subsequent
low-degree analyses and applications appeared in
\cite{Vid16,NV18a,NV18b,NW19}; the soundness analyses in
\cite{Vid16,NV18b} were later withdrawn because of an inherited proof error.
Ji, Natarajan, Vidick, Wright, and Yuen identified that error and established
quantum soundness of the classical low-individual-degree test by a different
analysis~\cite{JNVWY20LID}.  This corrected theorem supplies an important
ingredient for \(\mathsf{MIP}^*=\mathsf{RE}\)~\cite{JNVWY20MIP}.

The low-individual-degree test contains three types of checks.  The
axis-parallel lines test compares a value reported at a point with a
degree-\(d\) polynomial reported on a coordinate-parallel line through that
point.  The self-consistency test compares the two provers' answers at the
same point.  Finally, the diagonal-lines test compares points lying on more
general lines.  In the quantum analysis, this last component supplies a
compatibility guarantee between point measurements associated with arbitrary
pairs of points.  Such compatibility is automatic for a classical global
function, but not for quantum measurements.

This raises a natural question: can the diagonal-lines test be removed?  In
other words, does perfect, or nearly perfect, consistency on all
axis-parallel lines already force the point measurements to approximately
commute globally?  A counterexample for
\((m,d,q)=(2,2,4)\) was mentioned in~\cite{JNVWY20LID}.  However, that example
does not address the large-field regime relevant to low-degree soundness,
where \(d/q\) tends to zero.  The purpose of this paper is to show that the
same obstruction persists in that regime.

The row-and-column compatibility pattern has a well-known precedent in the
Mermin--Peres magic square and its nonlocal-game formulation
\cite{Mermin90,Peres90,Aravind02,CHTW04}: observables belonging to a common
row or column are compatible, while observables in different rows and columns
need not commute.  Ji, Natarajan, Vidick, Wright, and Yuen explicitly used
this comparison when formulating the large-field question answered here
\cite{JNVWY20LID}.  Their \(q=4\) construction and the magic-square examples
do not themselves yield the regime \(d/q\to0\).

For every odd prime \(q\), we construct a two-prover projective strategy for
the \((m,d)=(2,2)\) axis-parallel test over \(\mathbb F_q^2\).  The strategy
passes both orientations of every axis-parallel line-point check, as well as
the same-point consistency check, with probability one.  Nevertheless, for
independent uniform points \(\bu,\bv\in\mathbb F_q^2\), its point
measurements satisfy the exact identity
\[
\mathbb E_{\bu,\bv}\sum_{a,b\in\mathbb F_q}
\bigl\|([A_a^{\bu},A_b^{\bv}]\otimes I)\ket\Phi\bigr\|^2
=2\left(1-\frac1q\right)^3.
\]
Thus the average state-dependent commutator error approaches \(2\), rather
than zero, even though the retained tests have zero error and
\(d/q=2/q\to0\).

We also show that this phenomenon is not merely a failure of one particular
commutator argument.  Let \(G=\{G_g\}_g\) be any projective measurement whose
outcomes are functions \(g:\mathbb F_q^2\to\mathbb F_q\).  The point
measurements of our strategy have consistency with \(G\) at most
\[
1-\frac1{16}\left(1-\frac1q\right)^3.
\]
Since every odd prime satisfies \(q\geq3\), this is at most \(1-1/54\),
uniformly in \(q\).  In particular, the strategy remains a constant distance
from every global low-individual-degree polynomial measurement.

The construction uses standard harmonic-analysis and quantum-information
machinery.  The Fourier projections used to build the point and line
measurements come from harmonic analysis on finite abelian
groups~\cite{Rudin62}.  The use of a maximally entangled state together with
transposed measurements is standard in quantum information~\cite{Watrous18}.
Finite unitary operator bases and the finite-field Weyl--Heisenberg system go
back at least to Schwinger~\cite{Schwinger60}; their symplectic phase-space
formulation is standard in finite-dimensional quantum
mechanics~\cite{Wootters87,Gross06}.  The problem-specific ingredient of our
construction is the following explicit label.  To each point
\((x,y)\in\mathbb F_q^2\), we associate
\[
w(x,y)=(1,y,y^2;\,xy^2,-2xy,x),
\]
which obeys
\[
\langle w(x,y),w(x',y')\rangle
=(x'-x)(y'-y)^2.
\]
Consequently, the corresponding point measurements commute along every row
and column but fail to commute whenever both coordinates differ.  On each
axis-parallel line, the labels span an isotropic subspace.  Its joint Fourier
measurement produces legal degree-at-most-two line answers and exactly
refines the point measurements on that line.  Weyl trace orthogonality then
yields the exact commutator value above.

Our conclusion is that axis-parallel consistency alone cannot enforce the
global compatibility required by the usual quantum soundness conclusion.
The diagonal-lines test is one way to enforce the missing compatibility;
our result does not rule out replacing it with another explicit
compatibility test.

\paragraph{Organization of the paper.}
Section~\ref{sec:low-individual-degree-test} recalls the retained tests, and
Section~\ref{sec:main-result} states the main result.
Sections~\ref{sec:construction} and~\ref{sec:perfect-acceptance} construct the
strategy and prove perfect acceptance. Section~\ref{sec:global-obstruction}
establishes the commutator and global-measurement obstructions, and
Section~\ref{sec:discussion} discusses extensions.

\section{Low-individual-degree test}
\label{sec:low-individual-degree-test}
We recall the two components of the low-individual-degree test used in this
paper.  The diagonal-lines component is omitted throughout.

\begin{definition}[Axis-parallel lines test]
Choose a player \(w\) uniformly from
\(\{\mathrm{Alice},\mathrm{Bob}\}\), and write \(\overline w\) for the
other player.  Independently sample a point \(\bu\) uniformly from
\(\F_q^m\) and a coordinate \(\bi\) uniformly from \(\{1,\ldots,m\}\).
Let
\[
\bell:=\{\bu+t e_{\bi}:t\in\F_q\}
\]
be the axis-parallel line through \(\bu\) in direction \(e_{\bi}\).  The
verifier sends \(\bell\) to player \(w\), who returns a univariate polynomial
\(\boldf:\bell\to\F_q\) of degree at most \(d\), and sends \(\bu\) to
player \(\overline w\), who returns \(\ba\in\F_q\).  The verifier accepts
exactly when
\[
\boldf(\bu)=\ba.
\]
\end{definition}

\begin{definition}[Self-consistency test]
Sample \(\bu\) uniformly from \(\F_q^m\) and send it to both players.
Alice returns \(\ba\in\F_q\), while Bob returns \(\bb\in\F_q\).  The
verifier accepts exactly when \(\ba=\bb\).
\end{definition}

\begin{definition}[Reduced low-individual-degree test]
The reduced test chooses the axis-parallel lines test or the
self-consistency test, each with probability \(1/2\).  The diagonal-lines test
is omitted.
\end{definition}

\begin{definition}[Projective strategy]
A projective strategy consists of finite-dimensional Hilbert spaces
\(\mathcal H_{\mathrm{Alice}}\) and \(\mathcal H_{\mathrm{Bob}}\), a shared
bipartite state
\[
\ket\psi\in
\mathcal H_{\mathrm{Alice}}\otimes\mathcal H_{\mathrm{Bob}},
\]
and measurement families
\[
A^{\mathrm{Alice}},\quad B^{\mathrm{Alice}},\quad
A^{\mathrm{Bob}},\quad B^{\mathrm{Bob}}.
\]
For each player \(w\in\{\mathrm{Alice},\mathrm{Bob}\}\), these families
have the following form:
\begin{itemize}
\item For every \(u\in\F_q^m\),
\(A^{w,u}=\{A^{w,u}_a\}_{a\in\F_q}\) is a projective measurement on
\(\mathcal H_w\).
\item For every axis-parallel line \(\ell\subseteq\F_q^m\),
\(B^{w,\ell}=\{B^{w,\ell}_f\}_f\) is a projective measurement on
\(\mathcal H_w\), with outcomes indexed by the univariate polynomials
\(f:\ell\to\F_q\) of degree at most \(d\).
\end{itemize}
\end{definition}

\section{Main Result}
\label{sec:main-result}

For background, we refer the reader to \cite{JNVWY20LID}.  In the notation of
that reference, its quantum soundness theorem informally states that, for any
projective strategy
\[
(\ket{\psi}, A^{\mathrm{Alice}}, B^{\mathrm{Alice}}, L^{\mathrm{Alice}},
A^{\mathrm{Bob}}, B^{\mathrm{Bob}}, L^{\mathrm{Bob}})
\]
passing the full test with probability at least \(1-\epsilon\), there exists a
projective measurement \(G=\{G_g\}\) such that
\[
\mathbb{E}_{\bu \sim \F_q^m}\sum_{a \in \mathbb{F}_q}
\sum_{g:g(\bu)=a}
\bra{\psi} A^{\bu}_{a} \otimes G_g \ket{\psi}
\geq 1-\poly(m)\bigl(\poly(\epsilon)+\poly(d/q)\bigr).
\]
For fixed \(m\) and \(d\), the right-hand side tends to \(1\) as
\(\epsilon\to0\) and \(q\to\infty\).

The question addressed here is whether the diagonal-lines test is necessary.
We show that the reduced test does not have the corresponding quantum
soundness property.

\begin{theorem}[Failure without the diagonal-lines test]
For every odd prime \(q\), there is a two-prover projective strategy for the
reduced low-individual-degree test with parameters \((m,d)=(2,2)\) that is
accepted with probability \(1\).  Writing \(\ket{\psi}\) for its shared state
and \(A^u\) for Alice's point measurement, every projective measurement
\(G=\{G_g\}_g\) on Bob's Hilbert space whose outcomes are functions
\(g:\F_q^2\to\F_q\) satisfies
\[
\mathbb{E}_{\bu\sim\F_q^2}
\sum_{a\in\F_q}\sum_{g:g(\bu)=a}
\bra{\psi}A_a^{\bu}\otimes G_g\ket{\psi}
\leq
1-\frac1{16}\left(1-\frac1q\right)^3
\leq 1-\frac1{54}.
\]
\end{theorem}

We construct the point measurements
\((A^{\mathrm{Alice}},A^{\mathrm{Bob}})\), the line measurements
\((B^{\mathrm{Alice}},B^{\mathrm{Bob}})\), and the shared state
\(\ket{\psi}\) in the following section.

\section{Construction}
\label{sec:construction}

\subsection{Hilbert-space construction}
Let \(q\) be an odd prime number, fix \(\omega = e^{2\pi i/q}\). We choose the Hilbert space as 
\[\mathcal{H} = \mathbb{C}^{\mathbb{F}_q^3} := \operatorname{span}_{\mathbb{C}} \{\ket{z}: z \in \mathbb{F}_q^3\},\]
where \(\{\ket z:z\in\mathbb F_q^3\}\) is the standard orthonormal basis. Obviously, \(\dim \mathcal{H} = q^3\).

\subsection{Quantum-state choice}
We choose the quantum state \(\ket{\psi}\) to be the maximally entangled state,
which we also denote by \(\ket\Phi\):
\[\ket{\psi}=\ket\Phi := \frac{1}{\sqrt{\dim \mathcal{H}}} \sum_{z \in \mathbb{F}_q^3}\ket{z} \otimes \ket{z}.\]

\subsection{Point-measurement construction}
For any \(u \in \mathbb{F}_q^2\) and \(a\in \mathbb{F}_q\), we construct the point measurement as:
\[
A_a^u = \frac{1}{q} \sum_{t \in \mathbb{F}_q} \omega^{-ta}W(tw(u)).
\]

We explain the meaning of \(w\) and \(W\) in the above definition.

\begin{itemize}
\item Definition of \(w\):

For $(x, y) \in \mathbb{F}_q^2$, define

$$
\boxed{
w(x, y):=\left(1, y, y^2 ; x y^2,-2 x y, x\right) \in \mathbb{F}_q^3 \oplus \mathbb{F}_q^3}.
$$
The motivation for this choice will be explained after
Lemma~\ref{lem:point-label-symplectic-product}, where we compute its crucial
symplectic property.
\item Definition of \(W\):

For \(p\in \mathbb{F}_q^3\), define the generalized Pauli \emph{phase} operator
\[
Z(p)\ket z=\omega^{p\cdot z}\ket z,
\]

For \(r\in \mathbb{F}_q^3\), define the generalized Pauli \emph{shift} operator
\[
X(r)\ket z=\ket{z+r}.
\]

For \((p,r)\in \mathbb{F}_q^3\oplus\mathbb{F}_q^3\), define the associated Weyl operator
\[
W(p,r):=\omega^{-\frac12p\cdot r}Z(p)X(r).
\]
\end{itemize}

For Alice, we use \(A_a^{\mathrm{Alice},u} = A_a^u\), while for Bob, we use the transpose \(A_a^{\mathrm{Bob},u} = (A_a^u)^\mathrm{T}\).

Now we need to show that \(A^u = \{A_a^u\}_{a\in \mathbb{F}_q}\) is indeed a projective measurement. To do so, we need to first establish some special properties of the Weyl operator \(W\).

For \(v=(p,r)\) and \(v'=(p',r')\), define the standard symplectic form by
\[
\langle v,v'\rangle:=p\cdot r'-r\cdot p'.
\]

\begin{lemma}[Phase--shift commutation]
\label{lem:phase-shift-commutation}
For every \(p,r\in\mathbb F_q^3\),
\[
X(r)Z(p)=\omega^{-p\cdot r}Z(p)X(r).
\]
\end{lemma}

\begin{proof}
For every computational-basis vector \(\ket z\),
\begin{align*}
X(r)Z(p)\ket z
&=\omega^{p\cdot z}\ket{z+r},\\
Z(p)X(r)\ket z
&=\omega^{p\cdot(z+r)}\ket{z+r}
=\omega^{p\cdot r}\omega^{p\cdot z}\ket{z+r}.
\end{align*}
Therefore
\[
X(r)Z(p)\ket z
=\omega^{-p\cdot r}Z(p)X(r)\ket z.
\]
Since this holds for every basis vector \(\ket z\), the claimed operator
identity follows.
\end{proof}

\begin{lemma}[Weyl multiplication law]
\label{lem:weyl-multiplication-law}
For all \(v,v'\in\mathbb F_q^3\oplus\mathbb F_q^3\),
\[
W(v)W(v')
=\omega^{\frac12\langle v,v'\rangle}W(v+v').
\]
\end{lemma}

\begin{proof}
Write \(v=(p,r)\) and \(v'=(p',r')\).  Using
Lemma~\ref{lem:phase-shift-commutation}, together with
\(Z(p)Z(p')=Z(p+p')\) and \(X(r)X(r')=X(r+r')\), we obtain
\begin{align*}
W(p,r)W(p',r')
&=\omega^{-\frac12p\cdot r-\frac12p'\cdot r'}
  Z(p)X(r)Z(p')X(r')\\
&=\omega^{-\frac12p\cdot r-\frac12p'\cdot r'-p'\cdot r}
  Z(p+p')X(r+r').
\end{align*}
On the other hand,
\[
W(p+p',r+r')
=\omega^{-\frac12(p+p')\cdot(r+r')}Z(p+p')X(r+r').
\]
The ratio between the two phases is
\[
\omega^{\frac12(p\cdot r'-r\cdot p')}
=\omega^{\frac12\langle v,v'\rangle},
\]
which proves the identity.
\end{proof}

\begin{lemma}[Adjoint of a Weyl operator]
\label{lem:weyl-adjoint}
For every \(v\in\mathbb F_q^3\oplus\mathbb F_q^3\),
\[
W(v)^\dagger=W(-v).
\]
\end{lemma}

\begin{proof}
Let \(v=(p,r)\).  Since \(Z(p)^\dagger=Z(-p)\) and
\(X(r)^\dagger=X(-r)\), Lemma~\ref{lem:phase-shift-commutation}, applied
with \((p,r)\) replaced by \((-p,-r)\), gives
\begin{align*}
W(p,r)^\dagger
&=\omega^{\frac12p\cdot r}X(-r)Z(-p)\\
&=\omega^{\frac12p\cdot r}\omega^{-p\cdot r}
  Z(-p)X(-r)\\
&=\omega^{-\frac12p\cdot r}Z(-p)X(-r)\\
&=W(-p,-r).
\end{align*}
\end{proof}

\begin{lemma}[Multiplication along a point label]
\label{lem:point-label-weyl-multiplication}
For every point \(u\in\mathbb F_q^2\) and all \(s,t\in\mathbb F_q\),
\[
W(sw(u))W(tw(u))=W((s+t)w(u)).
\]
\end{lemma}

\begin{proof}
The symplectic form is alternating, since for every \(v=(p,r)\),
\[
\langle v,v\rangle=p\cdot r-r\cdot p=0.
\]
By bilinearity,
\[
\langle sw(u),tw(u)\rangle
=st\langle w(u),w(u)\rangle=0.
\]
Lemma~\ref{lem:weyl-multiplication-law} therefore gives
\begin{align*}
W(sw(u))W(tw(u))
&=\omega^{\frac12\langle sw(u),tw(u)\rangle}
  W((s+t)w(u))\\
&=W((s+t)w(u)).
\end{align*}
\end{proof}

\begin{lemma}[Character orthogonality on \(\mathbb F_q\)]
\label{lem:fq-character-orthogonality}
For every \(c\in\mathbb F_q\),
\[
\sum_{t\in\mathbb F_q}\omega^{tc}
=
\begin{cases}
q,&c=0,\\
0,&c\neq0.
\end{cases}
\]
\end{lemma}

\begin{proof}
If \(c=0\), every summand is \(1\), so the sum is \(q\).  If \(c\neq0\),
multiplication by \(c\) permutes \(\mathbb F_q\), and hence
\[
\sum_{t\in\mathbb F_q}\omega^{tc}
=\sum_{z\in\mathbb F_q}\omega^z=0,
\]
where the last equality is the sum of all \(q\)-th roots of unity.
\end{proof}

\begin{proposition}[Point measurements]
\label{prop:point-measurements}
For every \(u\in\mathbb F_q^2\), the family
\(A^u=\{A_a^u\}_{a\in\mathbb F_q}\) is a projective measurement on
\(\mathcal H\).  Explicitly, for all \(a,b\in\mathbb F_q\),
\[
(A_a^u)^\dagger=A_a^u,\qquad
A_a^uA_b^u=\delta_{a,b}A_a^u,\qquad
\sum_{a\in\mathbb F_q}A_a^u=I.
\]
\end{proposition}

\begin{proof}
\noindent\emph{Hermiticity.}
By Lemma~\ref{lem:weyl-adjoint}, we obtain
\begin{align*}
(A_a^u)^\dagger
&=\frac1q\sum_{t\in\mathbb F_q}
  \omega^{ta}W(-tw(u))\\
&=\frac1q\sum_{s\in\mathbb F_q}
  \omega^{-sa}W(sw(u))
=A_a^u.
\end{align*}
For the second equality, set \(s=-t\).  Negation is a bijection of
\(\mathbb F_q\), and hence summing over \(t\) is equivalent to summing over
\(s\).

\par\medskip\noindent\emph{Orthogonality.}
For \(a,b\in\mathbb F_q\),
Lemma~\ref{lem:point-label-weyl-multiplication} gives
\begin{align*}
A_a^uA_b^u
&=\frac1{q^2}\sum_{s,t\in\mathbb F_q}
  \omega^{-sa-tb}W((s+t)w(u))\\
&=\frac1{q^2}\sum_{k\in\mathbb F_q}
  \omega^{-kb}W(kw(u))
  \sum_{s\in\mathbb F_q}\omega^{s(b-a)}.
\end{align*}
By Lemma~\ref{lem:fq-character-orthogonality}, applied with \(c=b-a\),
\[
\sum_{s\in\mathbb F_q}\omega^{s(b-a)}
=q\,\delta_{a,b}.
\]
Consequently,
\[
A_a^uA_b^u
=\delta_{a,b}\frac1q\sum_{k\in\mathbb F_q}
  \omega^{-ka}W(kw(u))
=\delta_{a,b}A_a^u.
\]

\par\medskip\noindent\emph{Completeness.}
Interchanging the two sums gives
\begin{align*}
\sum_{a\in\mathbb F_q}A_a^u
&=\frac1q\sum_{t\in\mathbb F_q}
  \left(\sum_{a\in\mathbb F_q}\omega^{-ta}\right)W(tw(u))\\
&=W(0)=I,
\end{align*}
where the second equality follows from
Lemma~\ref{lem:fq-character-orthogonality}, applied with \(c=-t\).
This proves all three identities and hence the proposition.
\end{proof}

\subsection{Line-measurement construction}
Fix an axis-parallel line \(\ell\subseteq\mathbb F_q^2\).  First, define
the vector subspace
\[
S_\ell = \operatorname{span}_{\mathbb{F}_q} \{w(u) : u \in \ell\} \subseteq \mathbb{F}_q^3 \oplus \mathbb{F}_q^3.
\]

\begin{itemize}
	\item Definition of \(\lambda\):

Let \(S_\ell^*\) denote the dual space of \(S_\ell\), that is,
\[
S_\ell^*:=\operatorname{Hom}_{\mathbb F_q}(S_\ell,\mathbb F_q),
\]
the space of \(\mathbb F_q\)-linear functionals from \(S_\ell\) to
\(\mathbb F_q\).

\item Definition of \(h_{\lambda}^\ell\):

For each \(\lambda\in S_\ell^*\), define the induced line answer
\[
h_\lambda^\ell \colon \ell \to \mathbb{F}_q,
\qquad h_\lambda^\ell(u) := \lambda(w(u)).
\]
Indeed, 
\begin{itemize}
	\item[$\circ$] On a horizontal line:
	\[\lambda(w(x, y_0))=\lambda\left((1, y_0, y_0^2 ; 0,0,0)\right)+x\cdot\lambda\left((0,0,0 ; y_0^2,-2 y_0, 1)\right),\]
	which is affine in \(x\).
	\item[$\circ$] On a vertical line:
	\[\lambda(w(x_0, y))=\lambda\left((1, 0, 0 ; 0,0, x_0)\right) + y\cdot \lambda\left((0, 1, 0; 0, -2x_0, 0)\right) + y^2 \cdot \lambda\left((0,0,1; x_0,0,0)\right),\]
which is quadratic in \(y\).
\end{itemize}
In summary, we have shown that \(\deg h_\lambda^\ell \le 2\).

\item Definition of \(P_\lambda^\ell\):

Define
\[P_\lambda^{\ell}:=\frac{1}{\left|S_{\ell}\right|} \sum_{v \in S_{\ell}} \omega^{-\lambda(v)} W(v).\]

As we will show later, the family
\(\{P_\lambda^\ell\}_{\lambda\in S_\ell^*}\) is a projective measurement.
\end{itemize}

For every legal line answer \(g\colon\ell\to\mathbb F_q\) of degree at most
\(d=2\), define
\[
B_g^{\ell}:=\sum_{\lambda\in S_\ell^*:\,h_\lambda^\ell=g}
P_\lambda^{\ell}.
\]

For Alice, we use
\(B_g^{\mathrm{Alice},\ell}=B_g^\ell\), while for Bob, we use the transpose
\(B_g^{\mathrm{Bob},\ell}=(B_g^\ell)^\mathrm{T}\).

We next prove that these operators form a projective measurement.  The first
step is Lemma~\ref{lem:line-label-space-isotropic}, which shows that the phase
in Lemma~\ref{lem:weyl-multiplication-law} vanishes on \(S_\ell\).

\begin{lemma}[Symplectic product of point labels]
\label{lem:point-label-symplectic-product}
For all \((x,y),(x',y')\in\mathbb F_q^2\),
\[
\langle w(x,y),w(x',y')\rangle
=(x'-x)(y'-y)^2.
\]
\end{lemma}

\begin{proof}
Direct expansion gives
\begin{align*}
\langle w(x,y),w(x',y')\rangle
&=(1,y,y^2)\cdot(x'y'^2,-2x'y',x')\\
&\quad -(xy^2,-2xy,x)\cdot(1,y',y'^2)\\
&=x'(y'^2-2yy'+y^2)-x(y^2-2yy'+y'^2)\\
&=(x'-x)(y'-y)^2.
\end{align*}
\end{proof}

Now we are ready to explain the motivation behind the definition \(w(x, y):=\left(1, y, y^2 ; x y^2,-2 x y, x\right)\). In fact, the definition of \(w\) is not unique. What matters for the construction is
that \(w\) satisfies the following three properties:

\begin{itemize}
\item \emph{Desired symplectic pairing.}
For all \(u=(x,y)\) and \(v=(x',y')\),
\[
\langle w(u),w(v)\rangle
=K(u,v):=(x'-x)(y'-y)^2.
\]
This quantity vanishes exactly when \(u\) and \(v\) lie on a common
axis-parallel line, and it is nonzero when both coordinates differ.
Moreover,
\[
K(u,v)=-K(v,u),
\]
as required of an alternating symplectic pairing over a field of odd
characteristic. The square on \(y'-y\) is essential: without it, the
resulting expression would be symmetric rather than skew-symmetric.

\item \emph{Nonvanishing point labels.}
For every \((x,y)\in\mathbb F_q^2\),
\[
w(x,y)\neq 0.
\]

\item \emph{Low-degree restrictions to axis-parallel lines.}
The restriction of \(w(x,y)\) to either type of axis-parallel line has
coordinatewise degree at most \(2\). Consequently, for every
\(\lambda\in S_\ell^*\), the induced line answer
\[
h_\lambda^\ell(u):=\lambda(w(u))
\]
has degree at most \(2\).
\end{itemize}

Thus, the particular formula \(w(x, y):=\left(1, y, y^2 ; x y^2,-2 x y, x\right)\) is only one possible
choice; any map satisfying these three properties would suffice for the
construction.

\begin{lemma}[Isotropy of the line-label space]
\label{lem:line-label-space-isotropic}
For every axis-parallel line \(\ell\), the subspace \(S_\ell\) is isotropic:
\[
\langle v,v'\rangle=0
\qquad
\text{for all }v,v'\in S_\ell.
\]
Consequently,
\[
W(v)W(v')=W(v+v')
\qquad
\text{for all }v,v'\in S_\ell.
\]
\end{lemma}

\begin{proof}
Let \(u=(x,y)\) and \(u'=(x',y')\) be points on \(\ell\).  Then either
\(x=x'\) or \(y=y'\), so their symplectic product is zero by
Lemma~\ref{lem:point-label-symplectic-product}.  Since the vectors
\(w(u)\), \(u\in\ell\), span
\(S_\ell\), bilinearity implies that the symplectic product vanishes on all
of \(S_\ell\times S_\ell\).  Lemma~\ref{lem:weyl-multiplication-law} now gives
\[
W(v)W(v')
=\omega^{\frac12\langle v,v'\rangle}W(v+v')
=W(v+v').
\]
\end{proof}

We will also use Lemma~\ref{lem:line-character-orthogonality}.

\begin{lemma}[Character orthogonality on \(S_\ell\)]
\label{lem:line-character-orthogonality}
For every \(\nu\in S_\ell^*\),
\[
\sum_{v\in S_\ell}\omega^{\nu(v)}
=
\begin{cases}
|S_\ell|,&\nu=0,\\
0,&\nu\neq0.
\end{cases}
\]
Moreover, for every \(v\in S_\ell\),
\[
\sum_{\lambda\in S_\ell^*}\omega^{\lambda(v)}
=
\begin{cases}
|S_\ell|,&v=0,\\
0,&v\neq0.
\end{cases}
\]
\end{lemma}

\begin{proof}
We first prove the first summation.
\begin{itemize}
	\item When \(\nu = 0\): The first sum is clearly
	 \[\sum_{v\in S_\ell}\omega^{\nu(v)} = \sum_{v\in S_\ell}\omega^{0} = \sum_{v\in S_\ell}1 = |S_\ell|.\]
	\item When \(\nu \neq 0\): Choose \(v_1\in S_\ell\) such that
\(\nu(v_1)\neq0\), and extend \(v_1\) to a basis
\[
v_1,v_2,\ldots,v_k
\]
of \(S_\ell\). Every \(v\in S_\ell\) has a unique representation
\[
v=\sum_{j=1}^k t_jv_j,
\qquad t_1,\ldots,t_k\in\mathbb F_q.
\]
Therefore
\begin{align*}
\sum_{v\in S_\ell}\omega^{\nu(v)}
&=
\sum_{t_1,\ldots,t_k\in\mathbb F_q}
\omega^{\sum_{j=1}^k t_j\nu(v_j)}\\
&=
\prod_{j=1}^k
\left(
\sum_{t_j\in\mathbb F_q}
\omega^{t_j\nu(v_j)}
\right).
\end{align*}
The factor corresponding to \(j=1\) is zero by
Lemma~\ref{lem:fq-character-orthogonality}, applied with
\(c=\nu(v_1)\neq0\). Hence
\[
\sum_{v\in S_\ell}\omega^{\nu(v)}=0.
\]
\end{itemize}

We next prove the identity involving a fixed \(v\in S_\ell\).
\begin{itemize}
\item When \(v=0\):
Then every summand is \(1\), and
\[
\sum_{\lambda\in S_\ell^*}\omega^{\lambda(v)}
=|S_\ell^*|
=|S_\ell|.
\]

\item When \(v\neq0\): Let \(v_1 = v\), extend \(v_1\) to a basis
\[
v_1,v_2,\ldots,v_k
\]
of \(S_\ell\), and let
\[
\lambda_1,\lambda_2,\ldots,\lambda_k
\]
be the corresponding dual basis of \(S_\ell^*\). Every
\(\lambda\in S_\ell^*\) has a unique representation
\[
\lambda=\sum_{j=1}^k s_j\lambda_j,
\qquad s_1,\ldots,s_k\in\mathbb F_q.
\]
Since \(\lambda_j(v_1)=\delta_{j,1}\), we have \(\lambda(v)=s_1\).
Consequently,
\begin{align*}
\sum_{\lambda\in S_\ell^*}\omega^{\lambda(v)}
&=
\sum_{s_1,\ldots,s_k\in\mathbb F_q}\omega^{s_1}\\
&=
q^{k-1}\sum_{s_1\in\mathbb F_q}\omega^{s_1}\\
&=0,
\end{align*}
where the last equality follows from
Lemma~\ref{lem:fq-character-orthogonality}, applied with \(c=1\).
\end{itemize}

\end{proof}

Now we are ready to show that the family
\(\{P_\lambda^\ell\}_{\lambda\in S_\ell^*}\) is a projective measurement.

\begin{proposition}[Line measurements]
\label{prop:line-measurements}
For every axis-parallel line \(\ell\), the family
\(\{P_\lambda^\ell\}_{\lambda\in S_\ell^*}\) is a projective measurement on
\(\mathcal H\).  Explicitly, for all \(\lambda,\mu\in S_\ell^*\),
\[
(P_\lambda^\ell)^\dagger=P_\lambda^\ell,\qquad
P_\lambda^\ell P_\mu^\ell
=\delta_{\lambda,\mu}P_\lambda^\ell,\qquad
\sum_{\lambda\in S_\ell^*}P_\lambda^\ell=I.
\]
Consequently, grouping these projections according to
\(h_\lambda^\ell=g\) makes
\(\{B_g^\ell\}_{g:\ell\to\mathbb F_q,\,\deg(g)\leq 2}\) a projective
measurement.
\end{proposition}

\begin{proof}
\noindent\emph{Hermiticity.}
By Lemma~\ref{lem:weyl-adjoint}, linearity of \(\lambda\), and the change of
variables \(z=-v\), we obtain
\begin{align*}
(P_\lambda^\ell)^\dagger
&=\frac{1}{|S_\ell|}\sum_{v\in S_\ell}\omega^{\lambda(v)}W(-v)\\
&=\frac{1}{|S_\ell|}\sum_{z\in S_\ell}\omega^{-\lambda(z)}W(z)
=P_\lambda^\ell.
\end{align*}

\par\medskip\noindent\emph{Orthogonality.}
Lemma~\ref{lem:line-label-space-isotropic} implies that
\(W(v)W(v')=W(v+v')\) for \(v,v'\in S_\ell\).  Therefore
\begin{align*}
P_\lambda^\ell P_\mu^\ell
&=\frac{1}{|S_\ell|^2}\sum_{v,v'\in S_\ell}
  \omega^{-\lambda(v)-\mu(v')}W(v+v')\\
&=\frac{1}{|S_\ell|^2}\sum_{z\in S_\ell}\omega^{-\mu(z)}W(z)
  \sum_{v\in S_\ell}\omega^{(\mu-\lambda)(v)}\\
&=\delta_{\lambda,\mu}\frac{1}{|S_\ell|}
  \sum_{z\in S_\ell}\omega^{-\lambda(z)}W(z)\\
&=\delta_{\lambda,\mu}P_\lambda^\ell,
\end{align*}
where in the second line we set \(z=v+v'\), and in the third line we used
the first identity of Lemma~\ref{lem:line-character-orthogonality} with
\(\nu=\mu-\lambda\).  Taking \(\lambda=\mu\) shows that each
\(P_\lambda^\ell\) is a projection, while taking \(\lambda\neq\mu\) shows
that distinct projections are orthogonal.

\par\medskip\noindent\emph{Completeness.}
Interchanging the sums and applying the second identity of
Lemma~\ref{lem:line-character-orthogonality} to \(-v\) gives
\begin{align*}
\sum_{\lambda\in S_\ell^*}P_\lambda^\ell
&=\frac{1}{|S_\ell|}\sum_{v\in S_\ell}
  \left(\sum_{\lambda\in S_\ell^*}\omega^{-\lambda(v)}\right)W(v)\\
&=W(0)=I.
\end{align*}
This proves all three identities for the operators \(P_\lambda^\ell\).

Because every \(h_\lambda^\ell\) has degree at most \(2\), the fibers
\(\{\lambda:h_\lambda^\ell=g\}\) partition \(S_\ell^*\) as \(g\) ranges over
the legal line answers.  Grouping the mutually orthogonal projections within
each fiber therefore shows that
\[
B_g^\ell=\sum_{\lambda:h_\lambda^\ell=g}P_\lambda^\ell
\]
defines a projective measurement.
\end{proof}

\section{Perfect acceptance}
\label{sec:perfect-acceptance}
In this section, we prove that the constructed strategy passes the self-consistency test and the axis-parallel lines test with probability \(1\).

\subsection{Perfect self-consistency}
Given \(u\in \mathbb{F}_q^2\), suppose Alice's output is \(a\in \mathbb{F}_q\) and Bob's output is \(b \in \mathbb{F}_q\), the probability of this event is 
\[p_{\text{self}}(a, b) = \bra{\Phi}A_a^u\otimes(A_b^u)^{\mathrm{T}}\ket{\Phi}.\]

To proceed, we introduce the following special identity for maximally entangled state.

\begin{lemma}[Maximally entangled-state identity]
\label{lem:maximally-entangled-identity}
For all operators \(M,N\) on \(\mathcal H\),
\[
\bra\Phi M\otimes N\ket\Phi
=\frac{1}{\dim\mathcal H}\operatorname{Tr}(MN^{\mathrm T}).
\]
\end{lemma}

\begin{proof}
Writing matrix entries in the computational basis gives
\[
\bra\Phi M\otimes N\ket\Phi
=\frac{1}{\dim\mathcal H}\sum_{i,j}M_{ij}N_{ij}
=\frac{1}{\dim\mathcal H}\operatorname{Tr}(MN^{\mathrm T}).
\]
\end{proof}

By Lemma~\ref{lem:maximally-entangled-identity}, we have
\[
\begin{aligned}
p_{\text{self}}(a, b) &= \bra{\Phi}A_a^u\otimes(A_b^u)^{\mathrm{T}}\ket{\Phi}\\
&= \frac{1}{\dim \mathcal{H}}
   \operatorname{Tr}\!\left(A_a^u A_b^u\right)\\
&= \frac{\delta_{a,b}}{\dim \mathcal{H}}
   \operatorname{Tr}\!\left(A_a^u\right).
\end{aligned}
\]
In particular, \(p_{\mathrm{self}}(a,b)=0\) whenever \(a\neq b\). Therefore Alice and Bob
always give the same answer, so the strategy passes the self-consistency test
with probability \(1\).

\subsection{Perfect axis-parallel acceptance}
Given an axis-parallel line \(\ell \subseteq \mathbb{F}_q^2\) and \(u\in \ell\), suppose Alice's point measurement output is \(a\in \mathbb{F}_q\) and Bob's line measurement output is the polynomial \(g\colon \ell \to \mathbb{F}_q\). The probability of this event is
\[p_{\text{axis}}(a, g) = \bra{\Phi}A_a^u\otimes(B_g^\ell)^{\mathrm{T}}\ket{\Phi}.\]

To proceed, we need the following refinement lemma.
\begin{lemma}[Line measurements refine point measurements]
\label{lem:line-refinement}
For every axis-parallel line \(\ell\), \(u\in\ell\), and \(a\in \mathbb{F}_q\),
\[
  A^u_a
  =\sum_{g(u)=a}B^\ell_g.
\]
\end{lemma}

\begin{proof}
By the definitions of \(B_g^\ell\) and \(h_\lambda^\ell\), and because every
\(g,h_\lambda^\ell\) has degree at most \(2\),
\begin{align*}
\sum_{g(u)=a}B_g^\ell
&=\sum_{g(u)=a}\;
  \sum_{\lambda:h_\lambda^\ell=g}P_\lambda^\ell\\
&=\sum_{\lambda: h_\lambda^\ell(u)=a}P_\lambda^\ell\\
&=\sum_{\lambda:\lambda(w(u))=a}P_\lambda^\ell\\
&=\frac1{|S_\ell|}\sum_{v\in S_\ell}
 \left(\sum_{\lambda:\lambda(w(u))=a}\omega^{-\lambda(v)}\right)W(v)
\end{align*}
We will show the following identity:
\[
\sum_{\lambda:\lambda(w(u))=a}\omega^{-\lambda(v)}
=
\begin{cases}
\dfrac{|S_\ell|}{q}\,\omega^{-ta},&v=tw(u),\\[2mm]
0,&v\notin\{t w(u):t\in\mathbb F_q\}.
\end{cases}
\]

\begin{itemize}
	\item When \(v=tw(u)\) for some \(t\in \mathbb{F}_q\), the sum is
	\[
	\sum_{\lambda:\lambda(w(u))=a}\omega^{-\lambda(v)} = \sum_{\lambda:\lambda(w(u))=a}\omega^{-t\lambda(w(u))} = |\{\lambda:\lambda(w(u))=a\}|\cdot \omega^{-ta}
	\]
	We now count the linear functionals \(\lambda\) satisfying
	\(\lambda(w(u))=a\).  Write
	\[
	k:=\dim_{\mathbb F_q}S_\ell.
	\]
	The first coordinate of \(w(u)\) is \(1\), so \(w(u)\neq0\).  We can
	therefore choose a basis of \(S_\ell\) of the form
	\[
	w(u),e_2,\ldots,e_k.
	\]
	Let \(\{\phi_1,\ldots,\phi_k\}\) be the dual basis of
	\(\{w(u),e_2,\ldots,e_k\}\).  Then
	\[
	\lambda(w(u))=a
	\quad\Longleftrightarrow\quad
	\lambda=a\phi_1+\sum_{j=2}^k c_j\phi_j,
	\qquad c_j\in\mathbb F_q.
	\]
	Hence there are \(q^{k-1}\) such functionals, and
	\[
	q^{k-1}=\frac{|S_\ell|}{q}.
	\]
	It follows that
	\[
	\sum_{\lambda:\lambda(w(u))=a}\omega^{-\lambda(v)}
	=\frac{|S_\ell|}{q}\,\omega^{-ta}.
	\]

		\item When \(v\notin\{t w(u):t\in\mathbb F_q\}\), the vectors \(w(u)\) and
		\(v\) are linearly independent.  Extend \(\{w(u),v\}\) to a basis of
		\(S_\ell\), and let \(\{\phi_1,\ldots,\phi_k\}\) be its dual basis.
		Then every \(\lambda\) with \(\lambda(w(u))=a\) has the form
		\[
		\lambda=a\phi_1+\sum_{j=2}^k c_j\phi_j,
		\]
		and \(\lambda(v)=c_2\).  Therefore
		\[
		\sum_{\lambda:\lambda(w(u))=a}\omega^{-\lambda(v)}
		=q^{k-2}\sum_{c_2\in\mathbb F_q}\omega^{-c_2}=0.
		\]
	\end{itemize}

Finally, substituting the identity into the preceding sum proves
\begin{align*}
\frac1{|S_\ell|}\sum_{v\in S_\ell}
 \left(\sum_{\lambda:\lambda(w(u))=a}\omega^{-\lambda(v)}\right)W(v)
&=\frac1q\sum_{t\in\mathbb F_q}\omega^{-ta}W(tw(u))\\
&=A_a^u.
\end{align*}
\end{proof}

By Lemma~\ref{lem:line-refinement}, for every legal line answer \(g\),

\[
B_g^\ell\leq A^u_{g(u)}.
\]
Indeed, \(B_g^\ell\) is one of the mutually orthogonal summands in the
expression
\[
A^u_{g(u)}=\sum_{f(u)=g(u)}B_f^\ell.
\]
Consequently, if \(a\neq g(u)\), then the orthogonality of \(B_f^\ell\) gives
\(
A_a^uB_g^\ell=0.
\)

Using Lemma~\ref{lem:maximally-entangled-identity}, we therefore obtain
\begin{align*}
p_{\mathrm{axis}}(a,g)
&=\bra{\Phi}A_a^u\otimes(B_g^\ell)^{\mathrm T}\ket{\Phi}\\
&=\frac{1}{\dim\mathcal H}
  \operatorname{Tr}\!\left(A_a^uB_g^\ell\right)\\
&=0
\qquad\text{whenever }a\neq g(u).
\end{align*}
Thus every pair of answers that violates the verifier's condition
\(a=g(u)\) has probability zero.

The reverse orientation is identical.  If Alice receives the line and Bob
receives the point, then
\begin{align*}
p_{\mathrm{axis}}^{\mathrm{rev}}(g,a)
&=\bra{\Phi}B_g^\ell\otimes(A_a^u)^{\mathrm T}\ket{\Phi}\\
&=\frac{1}{\dim\mathcal H}
  \operatorname{Tr}\!\left(B_g^\ell A_a^u\right)\\
&=\frac{1}{\dim\mathcal H}
  \operatorname{Tr}\!\left(A_a^uB_g^\ell\right)\\
&=0
\qquad\text{whenever }a\neq g(u),
\end{align*}
Hence the strategy passes both orientations of the axis-parallel
line-point test with probability \(1\).

\section{Obstruction to global quantum soundness}
\label{sec:global-obstruction}

In this section, we prove the global-measurement obstruction in our main result: every global function-valued projective measurement has average point-consistency with the constructed measurements at most
\[
\mathbb{E}_{\bu \sim \F_q^2}\sum_{a \in \mathbb{F}_q} \sum_{g:g(\bu) = a} \bra{\psi} A^{\bu}_{a} \otimes G_g \ket{\psi}
\le 1-\frac1{54}.
\]Thus, this consistency is bounded away from \(1\) by a universal constant.

We will prove this in two steps. 
For \(u,v\in\mathbb F_q^2\), define
\[
\Gamma(u,v):=
\sum_{a,b\in\mathbb F_q}
\bigl\|([A_a^u,A_b^v]\otimes I)\ket\Phi\bigr\|^2.
\]

First, we will prove that
\[
\mathbb E_{\bu,\bv}\Gamma(\bu,\bv)
=2\left(1-\frac1q\right)^3.
\]

Second, we will prove that 
\[
\mathbb{E}_{\bu \sim \F_q^2}\sum_{a \in \mathbb{F}_q} \sum_{g:g(\bu) = a} \bra{\Phi} A^{\bu}_{a} \otimes G_g \ket{\Phi}
\le 1 - \frac1{32}\mathbb E_{\bu,\bv}\Gamma(\bu,\bv).
\]
Combining these two steps and the fact that \(q \ge 3\), the main result follows.

\subsection{Exact commutator calculation}
\label{sec:exact-commutator}

In this subsection, we prove the identity
\[
\mathbb E_{\bu,\bv}\Gamma(\bu,\bv)
=2\left(1-\frac1q\right)^3.
\]
To get a closed form of \(\Gamma(u,v)\), we will use the following
Lemma~\ref{lem:weyl-trace-orthogonality}.

\begin{lemma}[Weyl trace orthogonality]
\label{lem:weyl-trace-orthogonality}
For all \(v,v'\in\mathbb F_q^3\oplus\mathbb F_q^3\),
\[
\frac{1}{\dim\mathcal H}
\operatorname{Tr}\bigl(W(v')^\dagger W(v)\bigr)
=\delta_{v',v}.
\]
\end{lemma}

\begin{proof}
First consider a single Weyl operator \(W(p,r)\).  If \(r\neq0\), then
\(X(r)\ket z=\ket{z+r}\neq\ket z\) for every \(z\), so \(W(p,r)\) has no
nonzero diagonal entries and its trace is zero.  If \(r=0\), then
\[
\operatorname{Tr}(W(p,0))
=\operatorname{Tr}(Z(p))
=\sum_{z\in\mathbb F_q^3}\omega^{p\cdot z}.
\]
This equals \(\dim\mathcal H=q^3\) when \(p=0\), and it is zero when
\(p\neq0\) by Lemma~\ref{lem:fq-character-orthogonality}, applied to a
coordinate for which \(p\) is nonzero.

By Lemmas~\ref{lem:weyl-adjoint} and~\ref{lem:weyl-multiplication-law},
\[
W(v')^\dagger W(v)
=W(-v')W(v)
=\omega^{\frac12\langle-v',v\rangle}W(v-v').
\]
The preceding trace calculation is therefore zero unless \(v=v'\).  When
\(v=v'\), the operator is \(I\), whose trace is \(\dim\mathcal H\).
\end{proof}

\begin{theorem}[Exact pointwise commutator]
\label{thm:exact-pointwise-commutator}
Fix \(u,v\in\mathbb F_q^2\), and set \(c(u,v):=\langle w(u),w(v)\rangle\). Then we have
\[
\Gamma(u,v)=
\begin{cases}
0,
&c(u, v)=0,\\[2mm]
\dfrac{2(q-1)}q,
&c(u, v)\neq0.
\end{cases}
\]
\end{theorem}

\begin{proof}

Lemma~\ref{lem:weyl-multiplication-law}, applied in both orders, gives
\[
W(sw(u))W(tw(v))
=\omega^{st\,c(u,v)}W(tw(v))W(sw(u))
\]
for every \(s,t\in\mathbb F_q\).  
\begin{itemize}
	\item If \(c(u,v)=0\):
All the Weyl operators
appearing in \(A_a^u\) commute with those appearing in \(A_b^v\).  Hence
\([A_a^u,A_b^v]=0\) for every \(a,b\), and \(\Gamma(u,v)=\sum_{a,b\in\mathbb F_q}
\bigl\|([A_a^u,A_b^v]\otimes I)\ket\Phi\bigr\|^2=0\).

\item If \(c(u,v)\neq0\): For \(s,t\in\mathbb F_q\), define
\[
U_{s,t}:=W(sw(u))W(tw(v)),
\qquad
\alpha_{s,t}:=\omega^{-as-bt}
\left(1-\omega^{-st\,c(u,v)}\right).
\]
Substituting the defining character-sum formulas for the point measurements gives
\[
[A_a^u,A_b^v]
=\frac1{q^2}\sum_{s,t\in\mathbb F_q}
\alpha_{s,t}U_{s,t}.
\]
Since \(c(u, v) = \langle w(u), w(v)\rangle \neq 0\), the vectors \(w(u)\) and \(w(v)\) are linearly independent.  Therefore the
vectors
\[
sw(u)+tw(v),\qquad (s,t)\in\mathbb F_q^2,
\]
are all distinct.  By Lemma~\ref{lem:weyl-multiplication-law},
\(W(sw(u))W(tw(v))\) is a phase multiple of
\(W(sw(u)+tw(v))\).  Lemma~\ref{lem:weyl-trace-orthogonality}, together with
the distinctness above, gives
\[
\frac1{\dim\mathcal H}\operatorname{Tr}
\!\left(U_{s,t}^\dagger U_{s',t'}\right)
=\delta_{s,s'}\delta_{t,t'}.
\]
Then the preceding commutator formula and
Lemma~\ref{lem:maximally-entangled-identity} give
\[
\begin{aligned}
\bigl\|([A_a^u,A_b^v]\otimes I)\ket\Phi\bigr\|^2
&=\bra{\Phi}(([A_a^u,A_b^v])^\dagger([A_a^u,A_b^v])\otimes I) \ket{\Phi}\\
&=\frac1{\dim\mathcal H}\operatorname{Tr}\!\left(
  [A_a^u,A_b^v]^\dagger[A_a^u,A_b^v]\right)
  \qquad \text{(By Lemma~\ref{lem:maximally-entangled-identity})}\\
&=\frac1{\dim\mathcal H}\operatorname{Tr}\left(\frac1{q^4}\sum_{s,t,s',t'\in\mathbb F_q}
  \overline{\alpha_{s,t}}\alpha_{s',t'}
  U_{s,t}^\dagger U_{s',t'}\right)\\
&=\frac1{q^4}\sum_{s,t\in\mathbb F_q}|\alpha_{s,t}|^2\\
&=\frac1{q^4}\sum_{s,t\in\mathbb F_q}
\left|1-\omega^{-st\,c(u,v)}\right|^2.
\end{aligned}
\]
For each \(s\neq0\), the map \(t\mapsto st\,c(u,v)\) is a bijection of
\(\mathbb F_q\), while the terms with \(s=0\) vanish.  Moreover,
\begin{align*}
\sum_{z\in\mathbb F_q}|1-\omega^z|^2
&=\sum_{z\in\mathbb F_q}(2-\omega^z-\omega^{-z})\\
&=2q
\end{align*}
by Lemma~\ref{lem:fq-character-orthogonality}.  Hence
\[
\sum_{s,t\in\mathbb F_q}|1-\omega^{-st\,c(u,v)}|^2
=2q(q-1),
\]
and each pair \((a,b)\) contributes \(2(q-1)/q^3\).  Summing over the
\(q^2\) choices of \((a,b)\) gives
\[
\Gamma(u,v)=\frac{2(q-1)}q.
\]
\end{itemize}
\end{proof}

\begin{corollary}[Exact averaged commutator]
\label{cor:exact-average-commutator}
For independent uniformly random \(\bu,\bv\in\mathbb F_q^2\),
\[
\mathbb E_{\bu,\bv}\Gamma(\bu,\bv)
=2\left(1-\frac1q\right)^3.
\]
\end{corollary}

\begin{proof}
If \(u=(x,y)\) and \(v=(x',y')\), then
Lemma~\ref{lem:point-label-symplectic-product} gives
\[
\langle w(u),w(v)\rangle=(x'-x)(y'-y)^2.
\]
Thus \(\langle w(u),w(v)\rangle\neq0\) exactly when \(x\neq x'\) and
\(y\neq y'\), which occurs
with probability \((1-1/q)^2\) for independent uniform \(u,v\).  Therefore
\[
\mathbb E_{\bu,\bv}\Gamma(\bu,\bv)
=\frac{2(q-1)}q\left(1-\frac1q\right)^2
=2\left(1-\frac1q\right)^3.
\]
\end{proof}

\subsection{Constant separation}

In this subsection, we prove the final step, and get the main result of this paper.

\begin{theorem}[Commutator obstruction to global consistency]
\label{thm:global-measurement-obstruction}
For every projective measurement \(G=\{G_g\}_g\) on \(\mathcal H\), whose
outcomes are functions \(g:\mathbb F_q^2\to\mathbb F_q\),
\[
\mathbb E_{\bu\sim\mathbb F_q^2}
\sum_{a\in\mathbb F_q}\sum_{g:g(\bu)=a}
\bra\Phi A_a^{\bu}\otimes G_g\ket\Phi
\leq
1-\frac1{32}\mathbb E_{\bu,\bv}\Gamma(\bu,\bv).
\]
In particular, this holds for every measurement whose outcomes are
low-individual-degree polynomials.
\end{theorem}

\begin{proof}
Set
\[
G_a^u:=\sum_{g:g(u)=a}G_g.
\]
For every \(u\), \(\{G_a^u\}_a\) is a projective measurement.  Moreover,
all the \(G_a^u\)'s commute, since they are coarse-grainings of the single
projective measurement \(G\).  Consequently, their transposes are also
projective and commute with one another.

Write
\[
\Delta_a^u:=A_a^u-(G_a^u)^{\mathrm T}.
\]
Lemma~\ref{lem:maximally-entangled-identity} implies that, for every operator
\(M\) on \(\mathcal H\),
\[
\bigl\|(M\otimes I)\ket\Phi\bigr\|^2
=\frac1{\dim \mathcal{H}}\operatorname{Tr}(M^\dagger M).
\]
Using this identity and projectivity, we obtain
\begin{align*}
\sum_a\bigl\|(\Delta_a^u\otimes I)\ket\Phi\bigr\|^2
&=\frac1{\dim \mathcal{H}}\sum_a
  \operatorname{Tr}\!\left(\bigl(A_a^u-(G_a^u)^{\mathrm T}\bigr)^2\right)\\
&=\frac1{\dim \mathcal{H}}\sum_a\Bigl(
  \operatorname{Tr}\!\left((A_a^u)^2\right)
  +\operatorname{Tr}\!\left(((G_a^u)^{\mathrm T})^2\right)
  -\operatorname{Tr}\!\left(A_a^u(G_a^u)^{\mathrm T}\right)
  -\operatorname{Tr}\!\left((G_a^u)^{\mathrm T}A_a^u\right)
  \Bigr)\\
&=\frac1{\dim \mathcal{H}}\sum_a\Bigl(
  \operatorname{Tr}(A_a^u)
  +\operatorname{Tr}((G_a^u)^{\mathrm T})
  -2\operatorname{Tr}\!\left(A_a^u(G_a^u)^{\mathrm T}\right)
  \Bigr)\\
&=\frac1{\dim \mathcal{H}}\left(
  \operatorname{Tr}(I)+\operatorname{Tr}(I)
  -2\sum_a\operatorname{Tr}\!\left(A_a^u(G_a^u)^{\mathrm T}\right)
  \right)\\
&=2-\frac{2}{\dim \mathcal{H}}\sum_a
  \operatorname{Tr}\!\left(A_a^u(G_a^u)^{\mathrm T}\right)\\
&=2\left(1-\sum_a\bra\Phi A_a^u\otimes G_a^u\ket\Phi\right).
\end{align*}
Consequently,
\[
\mathbb E_u\sum_a
\bigl\|(\Delta_a^u\otimes I)\ket\Phi\bigr\|^2
=2\left(1-\mathbb E_u\sum_a
\bra\Phi A_a^u\otimes G_a^u\ket\Phi\right).
\]

Since the transposed global marginals commute,
\begin{align*}
[A_a^u,A_b^v]
&=\left(\Delta_a^u+(G_a^u)^{\mathrm T}\right)
  \left(\Delta_b^v+(G_b^v)^{\mathrm T}\right)
  -\left(\Delta_b^v+(G_b^v)^{\mathrm T}\right)
  \left(\Delta_a^u+(G_a^u)^{\mathrm T}\right)\\
&=\Delta_a^u\Delta_b^v
  +\Delta_a^u(G_b^v)^{\mathrm T}
  +(G_a^u)^{\mathrm T}\Delta_b^v
  -\Delta_b^v\Delta_a^u
  -\Delta_b^v(G_a^u)^{\mathrm T}
  -(G_b^v)^{\mathrm T}\Delta_a^u\\
&=\Delta_a^uA_b^v+(G_a^u)^{\mathrm T}\Delta_b^v
  -\Delta_b^vA_a^u-(G_b^v)^{\mathrm T}\Delta_a^u.
\end{align*}
The second equality uses the commutativity of the transposed global
marginals, namely,
\[
(G_a^u)^{\mathrm T}(G_b^v)^{\mathrm T}
=(G_b^v)^{\mathrm T}(G_a^u)^{\mathrm T}
\]

For any projective measurement \(R=\{R_j\}_j\), the preceding trace formula
and \(\sum_jR_j=I\) give
\[
\sum_j\bigl\|((XR_j)\otimes I)\ket\Phi\bigr\|^2
=\sum_j\bigl\|((R_jX)\otimes I)\ket\Phi\bigr\|^2
=\bigl\|(X\otimes I)\ket\Phi\bigr\|^2.
\]
For brevity, set
\[
D_u:=\sum_a\bigl\|(\Delta_a^u\otimes I)\ket\Phi\bigr\|^2,
\qquad
D_v:=\sum_b\bigl\|(\Delta_b^v\otimes I)\ket\Phi\bigr\|^2.
\]
We apply the preceding identity separately to each of the four terms in the
commutator expansion.  For the first and fourth terms, summing over \(b\)
uses the projective measurements \(A^v=\{A_b^v\}_b\) and
\(\{(G_b^v)^{\mathrm T}\}_b\), respectively.  Thus
\begin{align*}
\sum_{a,b}\bigl\|((\Delta_a^uA_b^v)\otimes I)\ket\Phi\bigr\|^2
&=\sum_a\bigl\|(\Delta_a^u\otimes I)\ket\Phi\bigr\|^2
=D_u,\\
\sum_{a,b}\bigl\|(((G_b^v)^{\mathrm T}\Delta_a^u)\otimes I)
\ket\Phi\bigr\|^2
&=\sum_a\bigl\|(\Delta_a^u\otimes I)\ket\Phi\bigr\|^2
=D_u.
\end{align*}
Similarly, for the second and third terms, summing over \(a\) uses the
projective measurements \(\{(G_a^u)^{\mathrm T}\}_a\) and
\(A^u=\{A_a^u\}_a\), respectively.  Hence
\begin{align*}
\sum_{a,b}\bigl\|(((G_a^u)^{\mathrm T}\Delta_b^v)\otimes I)
\ket\Phi\bigr\|^2
&=\sum_b\bigl\|(\Delta_b^v\otimes I)\ket\Phi\bigr\|^2
=D_v,\\
\sum_{a,b}\bigl\|((\Delta_b^vA_a^u)\otimes I)\ket\Phi\bigr\|^2
&=\sum_b\bigl\|(\Delta_b^v\otimes I)\ket\Phi\bigr\|^2
=D_v.
\end{align*}
Now regard each family of vectors indexed by \((a,b)\) as one vector in the
direct-sum Hilbert space
\(\bigoplus_{a,b}(\mathcal H\otimes\mathcal H)\).  The triangle inequality,
applied to the four families in the commutator expansion, gives
\begin{align*}
&\sqrt{
\sum_{a,b}\bigl\|([A_a^u,A_b^v]\otimes I)\ket\Phi\bigr\|^2
}\\
&\quad\leq
\sqrt{D_u}+\sqrt{D_v}+\sqrt{D_v}+\sqrt{D_u}\\
&\quad=2\left(\sqrt{D_u}+\sqrt{D_v}\right)\\
&\quad=2\left(
\sqrt{\sum_a\bigl\|(\Delta_a^u\otimes I)\ket\Phi\bigr\|^2}
+
\sqrt{\sum_b\bigl\|(\Delta_b^v\otimes I)\ket\Phi\bigr\|^2}
\right).
\end{align*}
Squaring, using \((\sqrt{x}+\sqrt{y})^2\leq2(x+y)\), and averaging over
independent uniform \(u,v\) gives
\[
\mathbb E_{u,v}\sum_{a,b}
\bigl\|([A_a^u,A_b^v]\otimes I)\ket\Phi\bigr\|^2
\leq
16\mathbb E_u\sum_a
\bigl\|(\Delta_a^u\otimes I)\ket\Phi\bigr\|^2
=32\left(1-\mathbb E_u\sum_a
\bra\Phi A_a^u\otimes G_a^u\ket\Phi\right).
\]
By the definition of \(\Gamma(u,v)\), rearranging gives
\[
\mathbb E_u\sum_a\bra\Phi A_a^u\otimes G_a^u\ket\Phi
\leq 1-\frac1{32}\mathbb E_{u,v}\Gamma(u,v),
\]
which is the claimed inequality.  Low-individual-degree polynomials are a
special class of functions \(\mathbb F_q^2\to\mathbb F_q\), which proves the
final statement.
\end{proof}

	\begin{corollary}[Explicit constant separation]
	For every projective measurement \(G=\{G_g\}_g\) as in
	Theorem~\ref{thm:global-measurement-obstruction},
	\[
	\mathbb E_{\bu\sim\mathbb F_q^2}
	\sum_{a\in\mathbb F_q}\sum_{g:g(\bu)=a}
	\bra\Phi A_a^{\bu}\otimes G_g\ket\Phi
	\leq
	1-\frac1{16}\left(1-\frac1q\right)^3
	\leq 1-\frac1{54}.
	\]
	\end{corollary}

	\begin{proof}
	Corollary~\ref{cor:exact-average-commutator} gives
	\[
	\frac1{32}\mathbb E_{\bu,\bv}\Gamma(\bu,\bv)
	=\frac1{16}\left(1-\frac1q\right)^3.
	\]
	Moreover, every odd prime satisfies \(q\geq3\), so
	\[
	\frac1{16}\left(1-\frac1q\right)^3
	\geq
	\frac1{16}\left(1-\frac13\right)^3
	=\frac1{54}.
	\]
	The result follows from
	Theorem~\ref{thm:global-measurement-obstruction}.
	\end{proof}

\section{Discussion}
\label{sec:discussion}
We mention that the above example can be generalized in three directions, though we do not give a rigorous treatment in this paper:
\begin{itemize}
	\item \(q\) can be chosen as any prime power: When \(q\) is odd prime power, the generalization is straightforward; when \(q = 2^k\), the generalization is more involved, but can be treated.
	\item \(m\) can be larger than 2: This is simply because we can embed the example into higher dimensional spaces.
	\item \(d\) can be larger than \(2\): This is simply because any lower degree polynomial is naturally a higher degree polynomial.
\end{itemize}

\bibliographystyle{alpha}
\bibliography{FormalLargeQ}

\end{document}